\documentclass{IEEEtran}

\usepackage{cite}

\usepackage{amsmath,amssymb,amsfonts,amsthm,mathtools}
\usepackage{booktabs}
\usepackage{graphicx}
\usepackage{url}

\usepackage[hidelinks]{hyperref}
\hypersetup{
  pdftitle={Universal Approximation of Maximal Lyapunov Functions with Anchored Neural Networks},
  pdfauthor={Jun Liu}
}

\newtheorem{theorem}{Theorem}
\newtheorem{lemma}[theorem]{Lemma}
\newtheorem{proposition}[theorem]{Proposition}
\newtheorem{corollary}[theorem]{Corollary}
\theoremstyle{definition}
\newtheorem{assumption}[theorem]{Assumption}
\theoremstyle{remark}
\newtheorem{remark}[theorem]{Remark}
\theoremstyle{definition}
\newtheorem{example}{Example}[section]

\newcommand{\R}{\mathbb{R}}
\newcommand{\A}{\mathcal{A}}
\newcommand{\Lie}{\mathcal{L}}
\newcommand{\norm}[1]{\left\lVert#1\right\rVert}
\newcommand{\set}[1]{\left\{#1\right\}}

\title{Universal Approximation of Maximal Lyapunov Functions with Anchored Neural Networks}

\author{Jun~Liu%
\thanks{This work was supported in part by the Natural Sciences and Engineering
Research Council of Canada and the Canada Research Chairs Program.}
\thanks{J. Liu is with the Department of Applied Mathematics, University of
Waterloo, Waterloo, ON N2L 3G1, Canada (e-mail: j.liu@uwaterloo.ca).}}

\begin{document}
\maketitle

\begin{abstract}
Maximal Lyapunov functions encode the entire domain of attraction of an
asymptotically stable equilibrium, but preserving strict decrease under
neural approximation is difficult because its margin vanishes at the
equilibrium. For systems locally dominated by an
asymptotically stable homogeneous vector
field, we construct a continuously differentiable maximal target and an
anchored, positivity-preserving neural family. We prove semiglobal
universal approximation: strict neural Lyapunov functions and their first
derivatives can approximate the target on nested invariant sublevel sets that
exhaust the domain of attraction. We also provide directly verifiable conditions under
which a candidate neural Lyapunov function can be formally certified, and
illustrate the effectiveness of the proposed neural architecture through
numerical examples.
\end{abstract}

\begin{IEEEkeywords}
Lyapunov functions, neural networks, nonlinear systems, region of attraction,
universal approximation, Zubov equation.
\end{IEEEkeywords}

\section{Introduction}

Lyapunov functions are a cornerstone of nonlinear stability analysis and
controller design.  A strict Lyapunov function certifies local asymptotic
stability, while a \emph{maximal} one is proper on the domain of attraction
(DOA) and has sublevel sets that exhaust it
\cite{Zubov1964,VannelliVidyasagar1985}.  Computational methods have long been
studied \cite{giesl2015review}.  Recent neural approaches directly approximate
Lyapunov functions \cite{Gruene2021}, compute Zubov solutions from data
\cite{kang2023data}, or use physics-informed neural networks (PINNs)
\cite{RaissiEtAl2019} to learn and verify maximal Lyapunov functions through
Zubov's equation \cite{liu2023towards,LiuEtAl2025}.  Despite this flexibility,
classical approximation theory alone does not yield neural certificates whose
sublevels exhaust the DOA.  Classical universal
approximation results for feedforward neural networks can control both
function values and first derivatives \cite{HornikEtAl1990}, yet two
obstacles prevent a direct application.  First, one must construct a
continuously differentiable maximal target compatible with a
positivity-preserving neural architecture.  Second, uniform approximation of
values and first derivatives does not preserve strict decrease near the
equilibrium because the decrease margin vanishes there.  We address these
issues directly.

Neural Lyapunov functions have been trained to enlarge data-informed safe
regions \cite{RichardsEtAl2018}.  For controlled systems, Zubov-inspired
sampling and iterative domain expansion have been combined with joint neural
controller--Lyapunov synthesis and GPU verification to enlarge certified
stability regions \cite{LiEtAl2025}.  Other methods exploit compositional
structure \cite{Gruene2021,GrueneSperl2023,sperl2023approximation} or enforce
positivity by construction
\cite{GabyEtAl2022}.  Counterexample-guided methods have combined neural
synthesis with falsification or verification based on satisfiability modulo
theories \cite{ChangEtAl2019,AbateEtAl2021,ZhouEtAl2022} and with
mixed-integer verification \cite{DaiEtAl2021}.  The Zubov--PINN framework
developed in \cite{LiuEtAl2025} learned and verified near-maximal DOA
estimates, but its
local verification step assumes exponential stability and uses a quadratic
Lyapunov function obtained from the linearization.  A related Taylor-neural
architecture combines a quadratic local approximation with neural
higher-order terms under the assumption of a Hurwitz linearization
\cite{BarreauBastianello2024}.  Subsequent work shows how to make a standard
feedforward network strictly verifiable without requiring a separate local
quadratic Lyapunov function, but it still assumes exponential stability
\cite{LiuFitzsimmons2025}.  The present work overcomes this limitation with a
novel neural network architecture for systems with an asymptotically stable
homogeneous leading part.

Our construction uses a single certified homogeneous Lyapunov function to
prescribe the local behavior of a continuously differentiable maximal target
and to anchor the neural architecture. We first establish universality of a
positivity-preserving neural parameterization for homogeneous Lyapunov
functions. Exact local matching then supplies the regularity needed to invoke
classical approximation results for functions and their first derivatives.
The anchored form enforces positivity, while homogeneous dominance guarantees
strict decrease near the origin for every finite network. On each prescribed
compact target sublevel, sufficiently accurate approximation of values and
first derivatives preserves strict decrease away from the origin and produces 
an invariant neural sublevel set. We show that a nested sequence of these
sublevel sets exhausts the DOA, thereby establishing semiglobal universality for
both function values and first derivatives. This is the main theoretical
contribution of the paper. We further derive directly verifiable conditions
under which a given neural sublevel is a formal Lyapunov certificate and
demonstrate them with numerical examples.

\section{Problem Formulation and Preliminaries}
\label{sec:setting}

Let $f\in C^1(\R^n;\R^n)$ satisfy $f(0)=0$, and let $\phi(t,x)$ denote its
maximal flow.  The DOA of the origin is
\[
 \begin{gathered}
 \A:=\set{x\in\R^n:\,\phi(t,x)\to0\text{ as }t\to\infty}.
 \end{gathered}
\]
For a differentiable function $V:\R^n\to\R$, we write
$\Lie_fV(x):=DV(x)f(x)$. Throughout, $\norm{\cdot}$ denotes the Euclidean
norm for vectors and the induced operator norm for linear maps, and
$\overline B_\rho:=\{x\in\R^n:\norm{x}\le\rho\}$. For a scalar function $h$
that is $C^1$ near a compact set $K$, write
\[
 \norm{h}_{C^1(K)}:=\sup_{x\in K}\bigl(|h(x)|+\norm{Dh(x)}\bigr).
\]
We call
$U\in C^1(\A;\R_{\ge0})$ a maximal Lyapunov function relative to $\A$ if it
is positive definite, satisfies $\Lie_fU<0$ on $\A\setminus\{0\}$, and is
proper on $\A$, meaning that $\{x\in\A:U(x)\le c\}$ is a compact subset of
$\A$ for every $c\ge0$.  These sublevels compactly exhaust $\A$ in the
following sense: they are nested compact subsets of $\A$, and their union is
$\A$. Our goal is to approximate such a function by a sequence of neural
Lyapunov functions whose invariant sublevel sets 
compactly exhaust $\A$.

A map $g:\R^n\to\R^m$ is homogeneous of degree $q\in \R$ if
$g(rx)=r^qg(x)$ for every $x\in\R^n$ and $r>0$.  

\begin{assumption}[Homogeneous dominance]
\label{ass:homogeneous}
There exist a real degree $d\ge1$, constants $\delta,L>0$, and a continuous
$d$-homogeneous vector field $f_d$ such that the origin of
$\dot x=f_d(x)$ is locally asymptotically stable and, in a neighborhood of
the origin,
\begin{align}
 f(x)&=f_d(x)+r_f(x),                                  \label{eq:expansion}\\
 \norm{r_f(x)}&\le L\norm{x}^{d+\delta}.              \nonumber
\end{align}
\end{assumption}

The following homogeneous converse Lyapunov theorem supplies the local Lyapunov
function used throughout the paper.

\begin{lemma}[Rosier's homogeneous converse Lyapunov theorem]
\label{lem:homogeneous-converse}
Under Assumption~\ref{ass:homogeneous}, for every $p>1$ there exist $\mu>0$
and a positive definite, proper, $p$-homogeneous function
\[
 V_p\in C^1(\R^n)\cap C^\infty(\R^n\setminus\{0\})
\]
such that
\begin{equation}
 DV_p(x)f_d(x)\le-\mu\norm{x}^{p+d-1},
 \qquad x\ne0.                                        \label{eq:hom-margin}
\end{equation}
\end{lemma}

\begin{proof}
Rosier's homogeneous converse theorem \cite[Theorem~2]{Rosier1992}, applied
with differentiability order one and homogeneity degree $p>1$, gives a
positive definite, proper $V_p\in C^1(\R^n)\cap
C^\infty(\R^n\setminus\{0\})$ satisfying $DV_p f_d<0$ away from the origin.
The function $DV_pf_d$ is homogeneous of degree $p+d-1$.  Its maximum on the
unit sphere is therefore negative.  Define
\[
 \mu:=-\max_{\norm{u}=1}DV_p(u)f_d(u)>0.
\]
Homogeneity then gives \eqref{eq:hom-margin}.
\end{proof}

Fix $p>1$ and one function $V_p$ furnished by
Lemma~\ref{lem:homogeneous-converse}.  Because $DV_p$ is homogeneous of
degree $p-1$, there is a constant $C_V>0$ such that
$\norm{DV_p(x)}\le C_V\norm{x}^{p-1}$.  Hence, near the origin,
\[
 \Lie_fV_p(x)
 \le -\mu\norm{x}^{p+d-1}
      +C_VL\norm{x}^{p+d-1+\delta}.
\]
Thus $V_p$ is a strict local Lyapunov function for the full vector field $f$.

The following classical universal approximation result controls both function
values and first derivatives \cite{HornikEtAl1990}.

\begin{lemma}[$C^1$ universal approximation]
\label{lem:derivative-uat}
Let $K\subset\R^n$ be compact, let $h\in C^1(\R^n)$, and let
$\varepsilon>0$.  There exist an integer $m\ge1$ and parameters
$A\in\R^{m\times n}$ and $b,w\in\R^m$ such that the network
\[
 N(x)=w^\top\tanh(Ax+b),
\]
where tanh acts componentwise, satisfies
\[
 \sup_{x\in K}\left(|N(x)-h(x)|+\norm{DN(x)-Dh(x)}\right)<\varepsilon.
\]
\end{lemma}

We next show that the positivity-preserving homogeneous architecture used in
the computations is universal.  For a width $m\ge1$, parameters
$A\in\R^{m\times n}$ and $b,w\in\R^m$, and $x=ru\ne0$ with
$r=\norm{x}$ and $\norm{u}=1$, define
\begin{equation}
 \widehat V(0)=0,\qquad
 \widehat V(ru)=r^p\exp\!\left(w^\top\tanh(Au+b)\right).
                                                               \label{eq:local-net}
\end{equation}

\begin{proposition}[Universality of the positive homogeneous architecture]
\label{prop:homogeneous-universality}
Let $V\in C^1(\R^n)$ be positive definite, proper, and $p$-homogeneous.  For
every $\varepsilon>0$, there exist an integer $m\ge1$ and parameters
$(A,b,w)$ such that the function \eqref{eq:local-net} satisfies
\begin{equation}
 \sup_{\norm{u}=1}
 \left(|\widehat V(u)-V(u)|+
       \norm{D\widehat V(u)-DV(u)}\right)<\varepsilon. \label{eq:homogeneous-C1}
\end{equation}
Every function of the form \eqref{eq:local-net} is positive definite, proper,
and $p$-homogeneous.  It belongs to
$C^1(\R^n)\cap C^\infty(\R^n\setminus\{0\})$ and has the origin as its unique
zero.  Moreover, if
\[
 DV(x)f_d(x)\le-\mu\norm{x}^{p+d-1},\qquad x\ne0,
\]
then the parameters can be chosen so that
\[
 D\widehat V(x)f_d(x)
 \le-\frac{\mu}{2}\norm{x}^{p+d-1},\qquad x\ne0.
\]
\end{proposition}

\begin{proof}
Set $h(u)=\log V(u)$ on the unit sphere.  A continuously differentiable
extension of $h$ to $\R^n$ can be obtained with a smooth radial cutoff on an
annulus containing the sphere; denote this extension again by $h$.
Lemma~\ref{lem:derivative-uat}, applied on the unit sphere, gives finite
parameters for which $N(u)=w^\top\tanh(Au+b)$ approximates $h$ and $Dh$
uniformly.  Writing $x=ru$, where $r=\norm{x}$, gives
$D\norm{x}|_{x=u}=u^\top$ and
$D(x/\norm{x})|_{x=u}=I-uu^\top$ for $\norm{u}=1$.  The chain rule therefore
gives
\begin{align*}
 D\log\widehat V(u)
 &=pu^\top+DN(u)(I-uu^\top),\\
 D\log V(u)
 &=pu^\top+Dh(u)(I-uu^\top).
\end{align*}
It follows that $\widehat V$ and $D\widehat V$ approximate $V$ and $DV$
uniformly on the sphere.  This proves \eqref{eq:homogeneous-C1}.

The exponential is positive and bounded above and below on the sphere.
Homogeneous extension therefore gives positive definiteness and properness,
while $p>1$ gives continuous differentiability at the origin.  Finally, let
$M_f=\max_{\norm{u}=1}\norm{f_d(u)}>0$.  Choosing the derivative error smaller
than $\mu/(2M_f)$ preserves the decrease margin on the sphere.  Homogeneity
then gives the stated inequality on $\R^n\setminus\{0\}$.
\end{proof}

\begin{remark}
It is shown in
\cite[Theorem~3.1 and Corollary~3.3]{FitzsimmonsLiu2024} that homogeneous neural
Lyapunov functions exist for asymptotically stable homogeneous vector fields.
Proposition~\ref{prop:homogeneous-universality} 
similarly establishes universality of the exponential--tanh
parameterization, whose positivity, properness, and unique zero hold for every
parameter choice by construction.
\end{remark}

\begin{remark}
For computation, one can choose and freeze $(A,b)$ and fit only $w$.  On the unit
sphere, $D\log\widehat V(u)f_d(u)$ is affine in $w$, so sampled decrease
inequalities form a linear program in $w$. The sampled program is not a
certificate: the decrease inequality must still be verified over the entire
sphere.  A verified function of the form \eqref{eq:local-net} may then be used
as the fixed $V_p$ in the rest of the paper.  When $d=1$, differentiability of
$f$ and \eqref{eq:expansion} imply $f_d(x)=Df(0)x$.  The leading system is
then linear, and local asymptotic stability makes $Df(0)$ Hurwitz.  In this case,
the usual quadratic Lyapunov function provides another choice of
$V_p$ with $p=2$.
\end{remark}

\section{A Target Maximal Lyapunov Function for Homogeneously Dominated Vector Fields}
\label{sec:target}

In this section, we construct a continuously differentiable maximal Lyapunov
function tailored to the homogeneous local certificate $V_p$.  We blend the
local decrease rate $-\Lie_fV_p$ with an outer cost chosen to ensure
properness, and integrate the resulting positive cost along trajectories.
The construction yields a target $U$ that agrees exactly with $V_p$ near the
origin, so their logarithmic ratio is a continuously differentiable
correction on the entire DOA.

Choose $0<a<b$ sufficiently small that $\{V_p\le b\}\subset\A$ and
\begin{equation}
0<V_p(x)\le b \quad \Longrightarrow \quad \Lie_fV_p(x)<0.                 \label{eq:local-decrease}
\end{equation}
Let $\chi\in C^1([0,\infty);[0,1])$ equal one on $[0,a]$ and zero on
$[b,\infty)$.  Choose $q\in C^1(\R^n)$ that is positive on
$\{V_p\ge a\}$ and satisfies
$q(x)\ge c_q(1+\norm{f(x)}^2)$ on $\{V_p\ge b\}$ for some $c_q>0$.  Define
\begin{equation}
 \omega(x)=
 \chi(V_p(x))[-\Lie_fV_p(x)]
 +[1-\chi(V_p(x))]q(x).                                \label{eq:omega}
\end{equation}
The interpolation is made only inside the region where both terms are
positive.  Thus $\omega$ is continuous and positive definite.  The simple
choice $q(x)=1+\norm{f(x)}^2$ always suffices.  The numerical examples use a
different admissible choice obtained by adding the dominant-field cost
$-DV_pf_d$ to a positive multiple of this coercive term.
This choice improves numerical scaling while retaining the coercive bound
above.

\begin{theorem}[Locally matched maximal Lyapunov function]
\label{thm:target}
For $x\in\A$, define
\begin{equation}
 U(x)=\int_0^\infty\omega(\phi(t,x))\,dt.               \label{eq:U}
\end{equation}
Then $U$ is finite, $C^1$, positive definite, and proper relative to $\A$.
Moreover,
\begin{equation}
 U=V_p\quad\text{on }\{V_p\le a\},\qquad
 \Lie_fU=-\omega\quad\text{on }\A.                    \label{eq:target-id}
\end{equation}
Consequently, the function $W:\R^n\to[0,1]$ defined by
\[
 W(x)=
 \begin{cases}
  \tanh U(x),&x\in\A,\\
  1,&x\in\R^n\setminus\A,
 \end{cases}
\]
is continuous on $\R^n$ and $C^1$ on $\A$.
\end{theorem}

The proof is given in Appendix~\ref{app:target}.

The following corollary establishes the regularity of the logarithmic
correction used in the multiplicative neural architecture below.  It also
shows that $U$ exactly matches $V_p$ on $\{V_p\le a\}$.

\begin{corollary}[Regular logarithmic correction]
\label{cor:factor}
The function $F:\A\to\R$ defined by
\begin{equation}
 F(x)=\begin{cases}
 \log\bigl(U(x)/V_p(x)\bigr),& x\in\A\setminus\{0\},\\
 0,&x=0,
 \end{cases}                                             \label{eq:F}
\end{equation}
belongs to $C^1(\A)$ and is identically zero on $\{V_p\le a\}$.
\end{corollary}

The following example shows that the simpler cost $\omega_d=-DV_pf_d$
guarantees neither exact local matching nor properness in general.  Let
$f(x)=-x-x^5$, $f_d(x)=-x$, and $V_p(x)=x^2/2$.  Then
$\omega_d(x)=x^2$ is positive definite, while $\A=\R$ and
\[
 \int_0^\infty\omega_d(\phi(t,x))\,dt
 =\tfrac12\arctan(x^2)\longrightarrow\tfrac\pi4.
\]
Thus the resulting value function is bounded and not maximal.  The outer
term in \eqref{eq:omega} is one convenient general way to enforce properness;
it can be omitted when properness of the chosen cost is established
separately.

\section{Anchored Neural Architecture and Universal Approximation}
\label{sec:universal}

This section introduces a neural architecture anchored by the certified local
Lyapunov function $V_p$.  We first prove that, for every parameter choice,
$H_\theta$ and $W_\theta$ are strict Lyapunov functions near the origin.  We
then show that the proposed architecture can approximate both the maximal
target $U$ and its bounded transform $W$ on arbitrarily large compact
sublevels while producing invariant sublevels that exhaust $\A$. 

Let $N_\theta:\R^n\to\R$ be a fully connected feedforward network with tanh activations, and define
\begin{equation}
 \begin{aligned}
 G_\theta(x)&=N_\theta(x)-N_\theta(0),\\
 H_\theta(x)&=V_p(x)e^{G_\theta(x)},\qquad
 W_\theta(x)=\tanh H_\theta(x).
 \end{aligned}                                         \label{eq:architecture}
\end{equation}
We refer to $H_\theta$ and $W_\theta$ as \emph{anchored} neural networks.  The fixed
factor $V_p$ serves as their anchor, while subtracting $N_\theta(0)$ centers
the neural correction at the origin.  Thus $G_\theta(0)=0$ and
$H_\theta(x)/V_p(x)\to1$ as $x\to0$.  Because $N_\theta$ is bounded,
there exist constants $0<m_\theta\le M_\theta<\infty$ such that
\[
 m_\theta V_p(x)\le H_\theta(x)\le M_\theta V_p(x).
\]
Thus $H_\theta$ is positive definite, with the origin as its unique zero.
Moreover, each sublevel $\{H_\theta\le c\}$ is a closed subset of the compact
set $\{V_p\le c/m_\theta\}$, so $H_\theta$ is proper.  The transformation
$W_\theta=\tanh H_\theta$ provides a bounded Zubov-type representation.
Since $H_\theta$ is finite and nonnegative, $0\le W_\theta(x)<1$ for every
$x$.  The strict monotonicity of tanh preserves sublevel sets up to a change
of level, while
\[
 \Lie_fW_\theta=(1-W_\theta^2)\Lie_fH_\theta
\]
preserves strict decrease.  Also,
$W_\theta(x)\to1$ as $\norm{x}\to\infty$.

The first result shows that both representations satisfy the local Lyapunov
conditions for every parameter vector.

\begin{proposition}[Automatic satisfaction of local Lyapunov conditions]
\label{prop:local}
Under Assumption~\ref{ass:homogeneous}, for every parameter vector $\theta$,
the corresponding functions $H_\theta$ and $W_\theta$ are strict Lyapunov
functions on some neighborhood of the origin.
\end{proposition}

\begin{proof}
There exist $\alpha,\rho_0>0$ such that
\[
 \Lie_fV_p(x)\le-\alpha\norm{x}^{p+d-1},
 \qquad 0<\norm{x}\le\rho_0.
\]
By homogeneity of $V_p$ and the local expansion of $f$, after decreasing
$\rho_0$ if necessary, there exist $M_V,M_f>0$ such that, for
$\norm{x}\le\rho_0$,
\[
 V_p(x)\le M_V\norm{x}^p,\qquad
 \norm{f(x)}\le M_f\norm{x}^d.
\]
Because $N_\theta$ is smooth, $DG_\theta$ is bounded on the closed ball
$\{x:\norm{x}\le\rho_0\}$; let $L_\theta$ be such a bound.  Differentiating
$H_\theta$ gives
\begin{equation}
 \Lie_fH_\theta=e^{G_\theta}
 \left(\Lie_fV_p+V_pDG_\theta f\right).                 \label{eq:Hdot}
\end{equation}
Consequently,
\[
 \Lie_fH_\theta(x)
 \le e^{G_\theta(x)}\norm{x}^{p+d-1}
 \left(-\alpha+M_VM_fL_\theta\norm{x}\right).
\]
The right-hand side is negative for all nonzero $x$ in a sufficiently small
ball.  Thus $H_\theta$ is a strict Lyapunov function there.  Since
$W_\theta=\tanh H_\theta$ is positive definite, $0\le W_\theta<1$, and
$\Lie_fW_\theta=(1-W_\theta^2)\Lie_fH_\theta$, the same conclusion holds for
$W_\theta$.
\end{proof}

The next result goes beyond this local guarantee: a suitable network can be
chosen so that both representations approximate their targets arbitrarily
well and are strict Lyapunov functions on any prescribed compact target
sublevel.

\begin{theorem}[Semiglobal universality of anchored neural Lyapunov functions]
\label{thm:semiglobal}
Let $U$ and $W$ be the target functions from
Theorem~\ref{thm:target}, and write
$K_\ell=\{x\in\A:U(x)\le\ell\}$.  For every
$0<\ell_-<\ell_+$ and every $\varepsilon>0$, there exist a tanh network
architecture, a parameter vector $\theta$, and a level $\gamma>0$ such that
\begin{equation}
 \begin{aligned}
  \norm{H_\theta-U}_{C^1(K_{\ell_+})}&<\varepsilon,\\
  \norm{W_\theta-W}_{C^1(K_{\ell_+})}&<\varepsilon.
 \end{aligned}                                             \label{eq:sg-approx}
\end{equation}
Moreover,
\begin{equation}
 \begin{aligned}
  \Lie_fH_\theta(x)&<0,\\
  \Lie_fW_\theta(x)&<0,
 \end{aligned}
 \qquad x\in K_{\ell_+}\setminus\{0\}.                  \label{eq:sg-decrease}
\end{equation}
Set $c_\theta=\tanh\gamma$.  The origin-containing components of
$\{H_\theta\le\gamma\}$ and $\{W_\theta\le c_\theta\}$ coincide; denote
their common value by $\Omega_{\theta,\gamma}^0$.  It satisfies
\begin{equation}
 K_{\ell_-}\subset\Omega_{\theta,\gamma}^0
 \subset\operatorname{int}K_{\ell_+}.                       \label{eq:sg-sandwich}
\end{equation}
Consequently, $\Omega_{\theta,\gamma}^0$ is compactly contained in $\A$, and
$H_\theta$ and $W_\theta$ are strict Lyapunov functions on this component.
Moreover, there exist tanh networks $N_{\theta_j}$ and levels
$\gamma_j>0$ whose corresponding components $\Omega_{\theta_j,\gamma_j}^0$
form a nested sequence with union $\A$.
\end{theorem}

The proof is given in Appendix~\ref{app:semiglobal}.  It uses the identity
$U=V_p$ near the origin and $C^1$ approximation on the remaining compact
annulus.

\section{Verification}
\label{sec:verification}

In this section, we discuss formal verification of a candidate neural network Lyapunov function given by the architecture (\ref{eq:architecture}). Given a certified $V_p$, the architecture (\ref{eq:architecture}) already makes $W_\theta$ positive definite with a unique zero.  It remains to prove strict 
decrease and containment of sublevel sets in a verification domain. To address the vanishing margin in Lyapunov inequalities near the origin, we separate a neighborhood of the origin from the remainder of the candidate sublevel: a homogeneous perturbation bound handles the inner neighborhood, while a nonlinear solver checks the outer shell and box boundary.

For $x=ru\ne0$, where $r=\norm{x}$ and $u=x/r$, set
$E(r,u):=r^{-d}f(ru)-f_d(u)$. Define
\begin{equation}
 \Phi_\theta(x)=r^{1-d}D\log H_\theta(x)f(x).             \label{eq:scaled-decrease}
\end{equation}
Then $f(ru)=r^d[f_d(u)+E(r,u)]$, and Assumption~\ref{ass:homogeneous}
gives $\norm{E(r,u)}\le Lr^\delta$ for sufficiently small $r$. Moreover,
\begin{align}
 \Phi_\theta(ru)
 &= [q_p(u)+rDG_\theta(ru)][f_d(u)+E(r,u)],\nonumber\\
 q_p(u)&=D\log V_p(u).                                   \label{eq:scaled-identity}
\end{align}

\begin{proposition}[Verifiable sublevel conditions]
\label{prop:verification}
Under Assumption~\ref{ass:homogeneous}, fix $\rho>0$, and let $\mathcal D\subset \R^n$ 
be compact with $\overline B_\rho\subset\operatorname{int}\mathcal D$.
Suppose that $\overline B_\rho$ lies in the neighborhood where
\eqref{eq:expansion} holds and there exist $\eta>0$ and
$M_q,M_d,L_G\ge0$ such that, for
$\norm{u}=1$ and $0<r\le\rho$,
\begin{align}
 q_p(u)f_d(u)&\le-\eta,&
 \norm{q_p(u)}&\le M_q,& \norm{f_d(u)}&\le M_d,\nonumber\\
 \norm{E(r,u)}&\le Lr^\delta,&
 \norm{DG_\theta(ru)}&\le L_G,                          \label{eq:verify-bounds}
\end{align}
and
\begin{equation}
 \beta_\rho:=M_qL\rho^\delta+
 \rho L_G(M_d+L\rho^\delta)<\eta.                       \label{eq:local-verification}
\end{equation}
If, for some $c\in(0,1)$ and $\varepsilon>0$,
\begin{align}
 x\in\mathcal D,\ \norm{x}\ge\rho,\ W_\theta(x)\le c
 &\Longrightarrow \Phi_\theta(x)\le-\varepsilon,       \label{eq:shell-verification}\\
 x\in\partial\mathcal D&\Longrightarrow W_\theta(x)>c, \label{eq:boundary-verification}
\end{align}
then $\Omega_{\theta,c}=\{x\in\mathcal D:W_\theta(x)\le c\}$ is compactly
contained in $\operatorname{int}\mathcal D$, is positively invariant, and
every trajectory starting in it converges to the origin.
\end{proposition}

\begin{proof}
Set $\alpha_\rho=\eta-\beta_\rho>0$.  For $0<r\le\rho$,
\eqref{eq:scaled-identity} and \eqref{eq:verify-bounds} give
\[
 \begin{aligned}
 &|\Phi_\theta(ru)-q_p(u)f_d(u)|
 \le M_qLr^\delta+rL_G(M_d+Lr^\delta)\le\beta_\rho.
 \end{aligned}
\]
Hence $\Phi_\theta(ru)\le-\alpha_\rho<0$ on the punctured ball.
Condition~\eqref{eq:shell-verification} gives
$\Phi_\theta\le-\varepsilon$ on the remainder of
$\Omega_{\theta,c}$.  Therefore $\Phi_\theta<0$ throughout
$\Omega_{\theta,c}\setminus\{0\}$.  For $x=ru\ne0$,
\[
 \Lie_fW_\theta(x)
 =(1-W_\theta(x)^2)H_\theta(x)r^{d-1}\Phi_\theta(x)<0,
\]
because every factor preceding $\Phi_\theta(x)$ is positive.

Continuity of $W_\theta$ makes $\Omega_{\theta,c}$ closed in the compact set
$\mathcal D$.  Condition~\eqref{eq:boundary-verification} makes it disjoint
from $\partial\mathcal D$; hence
$\Omega_{\theta,c}\Subset\operatorname{int}\mathcal D$.  Along a solution
starting in $\Omega_{\theta,c}$, $W_\theta$ cannot increase through the level
$c$, and the solution cannot reach $\partial\mathcal D$, where
$W_\theta>c$.  Thus $\Omega_{\theta,c}$ is positively invariant.  Its
compactness also gives forward completeness.  Finally,
$\Lie_fW_\theta=0$ in $\Omega_{\theta,c}$ only at the origin, so standard Lyapunov analysis 
gives $\phi(t,x)\to0$ for every
$x\in\Omega_{\theta,c}$.
\end{proof}

The bounds in \eqref{eq:verify-bounds} can be obtained using exact-rational
Taylor arithmetic, outward-rounded matrix and spectral-norm bounds, and
coefficient bounds. The homogeneous decrease condition in
\eqref{eq:verify-bounds}, the outer-shell condition
\eqref{eq:shell-verification}, and the boundary condition
\eqref{eq:boundary-verification} can then be verified using
dReal~\cite{GaoEtAl2013}.  GPU-based neural-network bound propagation with
adaptive domain subdivision provides a complementary approach for verifying
such inequalities; see \cite{LiEtAl2026Tutorial} for a recent tutorial.  
In the next section, we formally verify neural Lyapunov functions for
numerical examples using both dReal and the neural-network verifier
$\alpha,\beta$-CROWN. Verification code
implementing both approaches is provided.

\section{Examples}
\label{sec:examples}

The numerical experiments were developed using an experimental version of
LyZNet~\cite{liu2024lyznet}; the trained artifacts and self-contained
verification scripts are available at 
\url{https://github.com/j49liu/anchored-lyapunov-networks}.

\begin{example}[Two-machine swing dynamics]
\label{ex:two-machine}
The dynamics and domain are
\begin{align}
 \dot x_1&=x_2,\nonumber\\
 \dot x_2&=-\tfrac12x_2-\sin(x_1+\tfrac\pi3)+\sin\tfrac\pi3,\nonumber\\
 \mathcal D&=[-2,3]\times[-3,1.5].                       \label{eq:machine}
\end{align}
The origin is an exponentially stable equilibrium. Although a quadratic local
Lyapunov function could be obtained from the Hurwitz linearization, we deliberately
use a degree-two homogeneous neural anchor to illustrate the proposed architecture.
The anchor has $48$ tanh features, and the multiplicative correction network has
two tanh hidden layers of width $32$. We formally verify that
\[
 \begin{gathered}
 \Omega_{\theta,c}=\{x\in\mathcal D:W_\theta(x)\le c\},\quad c=0.7059,
 \end{gathered}
\]
is a certified inner approximation of the region of attraction. The verified
sublevel set and the learned Lyapunov function are shown in
Fig.~\ref{fig:examples}.  The dReal run~\cite{GaoEtAl2013}, using 48 parallel
jobs, and the CROWN bound-propagation run implemented in
$\alpha,\beta$-CROWN~\cite{XuEtAl2021}, on an NVIDIA H100
GPU, take $1044.80$ and $20.32$ seconds, respectively.
\end{example}

\begin{example}[Cubic-core planar system]
\label{ex:cubic-core}

Let $r^2=x_1^2+x_2^2$ and
\begin{equation}
 f_3=\begin{bmatrix}
4x_1^3-x_1^2x_2-6x_1x_2^2-x_2^3\\
x_1^3+4x_1^2x_2+x_1x_2^2-6x_2^3
\end{bmatrix}.                                           \label{eq:cubic}
\end{equation}
It is shown in \cite{LiuFitzsimmons2026} that the origin of $\dot x=f_3(x)$ is globally
asymptotically stable, but that $f_3$ admits no polynomial Lyapunov function.
We consider $f=f_3+h_5+h_7$ on
$\mathcal D=[-1.5,1.5]^2$, where
\begin{align}
 h_5&=\tfrac15r^4(3+\sin x_1+\cos x_2)x,\nonumber\\
 h_7&=\tfrac1{20}r^6(3+\cos(x_1+x_2)+\sin(x_1-x_2))x.   \label{eq:perturb}
\end{align}
We use a degree-two homogeneous neural anchor with $48$ tanh features and a
multiplicative correction network with two tanh hidden layers of width $32$. 
We formally verify that
\[
 \begin{gathered}
 \Omega_{\theta,c}=\{x\in\mathcal D:W_\theta(x)\le c\},\quad c=0.04154,
 \end{gathered}
\]
is a certified inner approximation of the region of attraction.
The verified sublevel set and the learned Lyapunov
function are shown in Fig.~\ref{fig:examples}.  The dReal and $\alpha,\beta$-CROWN runs take
$346.19$ and
$53.42$ seconds, respectively.
\end{example}

\begin{figure*}[!t]
\centering
\includegraphics[width=.88\textwidth]{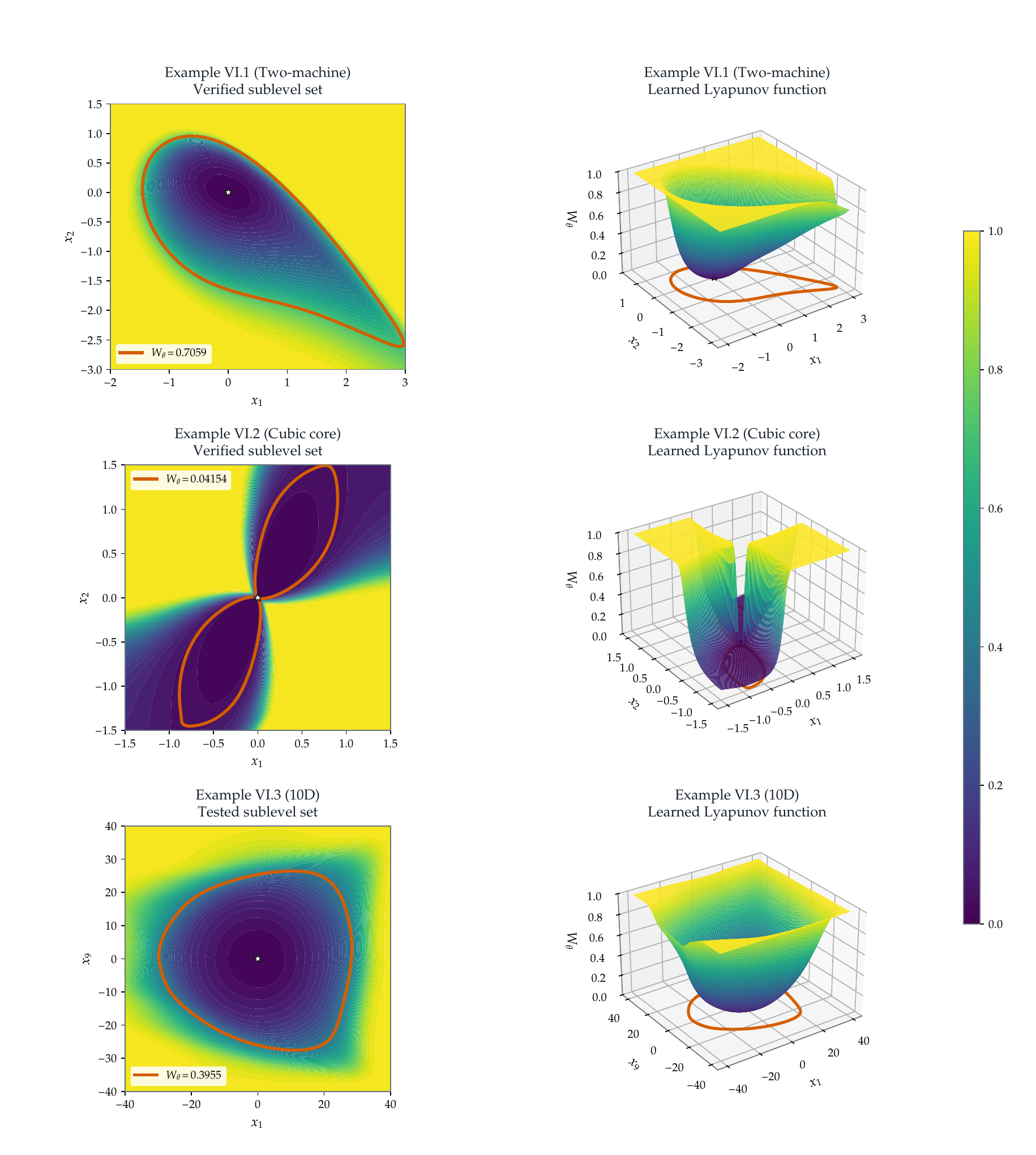}
\caption{
Sublevel sets and learned Lyapunov functions for
Examples~\ref{ex:two-machine}--\ref{ex:grune-10d}.  The first two
sublevel sets
are formally verified.  The 10D result is empirical only and is visualized on the non-invariant
$(x_1,x_9)$ coordinate slice with all other states zero. }
\label{fig:examples}
\end{figure*}

\begin{example}[10D system]
\label{ex:grune-10d}

We next consider the ten-dimensional system introduced by
Gr\"une~\cite{Gruene2021}, also studied in~\cite{LiuEtAl2025}.  Let
$z_i=(x_{2i-1},x_{2i})$ and
$A=\left[\begin{smallmatrix}-1&1/2\\-1/2&-1\end{smallmatrix}\right]$.
The system has the compact representation
\begin{equation}
 \dot z_i=Az_i+v_i,\qquad i=1,\ldots,5,                 \label{eq:grune10}
\end{equation}
where
\begin{align*}
 v_1&=(-.1x_9^2,0)^\top,&v_2&=(-.1x_1^2,0)^\top,\\
 v_3&=(.1x_7^2,0)^\top,&v_4&=0,&
 v_5&=(0,.1x_2^2)^\top.
\end{align*}
On $\mathcal D=[-40,40]^{10}$, we use a degree-two homogeneous neural anchor
with $64$ tanh features and a multiplicative correction network with three
tanh hidden layers of width $256$.  We test the 
sublevel set
\[
 \begin{gathered}
 \Omega_{\theta,c}=\{x\in\mathcal D:W_\theta(x)\le c\},\quad c=0.3955.
 \end{gathered}
\]
We draw $2^{20}$ IID points as
$x=40RZ/\norm{Z}_\infty$, where
$Z\sim\mathcal N(0,I_{10})$
and $R\sim\operatorname{Unif}[0,1]$ are
independent. Of these points, $50.0161\%$ lie in the sublevel, and every point
inside satisfies $\Phi_\theta<-10^{-2}$. A separate $2^{18}$-point test uses
Gaussian directions normalized in Euclidean norm and radii uniform on a
logarithmic scale from $10^{-12}$ to $10^{-3}$; all sampled points lie in the
sublevel and satisfy the same margin. Finally, none of $2^{18}$ sampled box-boundary
points lies in the sublevel. Figure~\ref{fig:examples} shows the result. We emphasize
that these tests provide numerical evidence rather than a formally verified
certificate. Verifying such certificates for high-dimensional systems remains
challenging for current SMT solvers such as dReal \cite{GaoEtAl2013} and
state-of-the-art neural network verifiers such as $\alpha,\beta$-CROWN
\cite{XuEtAl2021}. Designing neural Lyapunov
architectures that facilitate scalable verification is an interesting direction
for future work.
\end{example}

\section{Conclusion}

We establish universal approximation of maximal Lyapunov functions
for systems with an asymptotically stable homogeneous leading field.  A
locally matched cost extends a homogeneous local Lyapunov function to a
$C^1$ maximal target on the DOA.  The positive homogeneous architecture is
$C^1$-universal for the anchor class, while the multiplicative correction
preserves positivity and the dominant local asymptotics.  Every finite
anchored network is therefore strict on a parameter-dependent neighborhood;
on any prescribed compact target sublevel, sufficiently accurate members also
preserve decrease away from the origin. Nested neural sublevel sets can
thus exhaust the DOA.
We also provide verifiable sublevel conditions that can turn this structural result into a
computational certificate.

\emph{AI disclosure:} OpenAI Codex (GPT-5.6) assisted with coding and writing.
The author reviewed all AI-assisted material and takes full responsibility
for the content, claims, and any errors in this paper.

\appendices

\section{Proof of Theorem~\ref{thm:target}}
\label{app:target}

\begin{proof}
For $c\in\{a,b\}$, write
$P_c=\{V_p\le c\}$.  The set
$P_a$ is forward invariant by
\eqref{eq:local-decrease}.  If
$x\in P_a$, then
\[
 \int_0^\infty\omega(\phi(t,x))dt
 =-\int_0^\infty\frac{d}{dt}V_p(\phi(t,x))dt=V_p(x),
\]
and the strict Lyapunov argument gives
$V_p(\phi(t,x))\to0$, proving the exact inner identity.  Every $x\in\A$ reaches
$\operatorname{int}P_a$ in finite time, so the
integral before entrance is
finite and its tail is the terminal value $V_p$.

For regularity, set
\begin{equation}
 \zeta=(1-\chi(V_p))(q+\Lie_fV_p),
 \qquad \omega=-\Lie_fV_p+\zeta.                       \label{eq:zeta}
\end{equation}
The function $\zeta$ is $C^1$: it vanishes on an open neighborhood of the
origin, where $V_p$ need only be $C^1$, and all of its factors are $C^1$
away from the origin.
Fix $x_0\in\A$ and choose $T_0$ with
$\phi(T_0,x_0)\in\operatorname{int}P_a$.
Finite-time continuous dependence
and openness of the finite-time flow domain give a neighborhood $O$ of $x_0$
on which every solution exists through $T_0$ and
$\phi(T_0,x)\in\operatorname{int}P_a$.  Forward
invariance of $P_a$ then
gives $O\subset\A$.  Since $\zeta$ vanishes along every such trajectory after
$T_0$, on $O$ we have
\begin{equation}
 U(x)=V_p(x)+\int_0^{T_0}\zeta(\phi(t,x))\,dt.          \label{eq:fixed-time}
\end{equation}
The integrand $\zeta$ and the
finite-time flow are $C^1$, so differentiation under this finite integral
proves $U\in C^1(O)$.  The dynamic programming identity
\[
 U(\phi(t,x))=U(x)-\int_0^t\omega(\phi(s,x))ds
\]
then gives $\Lie_fU=-\omega$.

It remains to prove properness.  Continuity and positivity give a number
$m>0$ such that $\omega\ge m$ outside
$P_a$: on
$a\le V_p\le b$ this follows by compactness, and outside
$P_b$ it follows
from the outer cost.  Let $\tau_a$ and $\tau_b$ denote the first-entry times
into $P_a$ and
$P_b$, respectively.  For $U(x)\le M$,
$\tau_b\le\tau_a\le M/m$.  If
$x\notin P_b$, then
$\omega=q\ge c_q(1+\norm{f}^2)$ on
$[0,\tau_b]$, and therefore
\[
 \int_0^{\tau_b}\norm{f(\phi(t,x))}^2dt\le M/c_q.
\]
By Cauchy--Schwarz, the path length before $\tau_b$ is at most
$\sqrt{(M/m)(M/c_q)}$.  Both $P_a$ and
$P_b$ are forward invariant, so after
entrance into $P_b$ the trajectory remains in that
compact set.  Hence
$\{U\le M\}$ is bounded and all of its trajectory arcs through $T=M/m$ lie
in one common compact set.  If $x_j\to x$ in this sublevel, continuation and
continuous dependence through $T$ give
$\phi(T,x_j)\to\phi(T,x)\in P_a$.  Thus $x\in\A$,
and continuity of $U$ gives
$U(x)\le M$.  The sublevel is therefore closed, compact, and contained in
$\A$. If $x_j\in\A$ converges to a point of the ambient boundary
$\partial\A$, then $U(x_j)\to\infty$; otherwise, a subsequence would lie in one
compact sublevel of $U$ and have its limit in $\A$. Hence
$W(x_j)=\tanh U(x_j)\to1$, which proves continuity of the stated extension.
\end{proof}

\section{Proof of Theorem~\ref{thm:semiglobal}}
\label{app:semiglobal}

\begin{proof}
Fix $0<c<a$ so that
$P_c=\{V_p\le c\}$ lies inside the neighborhood where
the homogeneous expansion and remainder bound in~\eqref{eq:expansion} hold,
and $c<\ell_+$.  There are positive
constants $\alpha,C_V,C_f$ such that, throughout
$P_c$,
\begin{align*}
 \Lie_fV_p&\le-\alpha\norm{x}^{p+d-1},&
 V_p&\le C_V\norm{x}^p,\\
 \norm{f(x)}&\le C_f\norm{x}^d.&&
\end{align*}
Let $R_c=\max_{P_c}\norm{x}$.  The set
$Q=K_{\ell_+}\setminus\operatorname{int}P_c$ is a
compact annulus, so
\[
 m=\min_Q\omega>0,\qquad B=\max_Q\norm{f}<\infty.
\]

By Corollary~\ref{cor:factor}, $F=\log(U/V_p)$ belongs to $C^1(\A)$ and is
zero on $P_c$.  Choose
$\psi\in C_c^\infty(\A)$ with $\psi=1$ on a
neighborhood of $K_{\ell_+}$, and extend $\psi F$ by zero to $\R^n$.
The $C^1$ universal approximation result in
Lemma~\ref{lem:derivative-uat} gives tanh networks $N_j$ converging to this
extension together with their first derivatives.  Setting
$G_j=N_j-N_j(0)$ preserves derivative convergence and centers the correction
at the origin.  Hence
\[
 G_j\to F,\qquad H_j=V_pe^{G_j}\to U
 \quad\text{in }C^1(K_{\ell_+}),
\]
and $\sup_{P_c}\norm{DG_j}\to0$.  Since composition
by tanh is continuous
in $C^1$ on compact sets,
\[
 W_j=\tanh H_j\to\tanh U=W
 \quad\text{in }C^1(K_{\ell_+}).
\]

For all sufficiently large $j$, \eqref{eq:Hdot} gives, on
$P_c\setminus\{0\}$,
\[
 \Lie_fH_j
 \le e^{G_j}\!\left[-\alpha+C_VC_fR_c
       \sup_{P_c}\norm{DG_j}\right]
       \norm{x}^{p+d-1}<0.
\]
On the annulus $Q$,
\[
 \begin{aligned}
 \Lie_fH_j&=-\omega+D(H_j-U)f,\\
 &\le-m+B\sup_Q\norm{D(H_j-U)}<0.
 \end{aligned}
\]
The bound is strict for the same sufficiently accurate
approximation. Because $0\le W_j<1$
and $\Lie_fW_j=(1-W_j^2)\Lie_fH_j$, both strict-decrease inequalities in
\eqref{eq:sg-decrease} hold for all sufficiently large $j$.  Let
\[
 \tau=\min\!\left\{\frac{\varepsilon}{2},
                    \frac{\ell_+-\ell_-}{4}\right\}>0.
\]
Choose $j$ large enough that the preceding strict-decrease inequalities hold
and both $C^1$ errors in~\eqref{eq:sg-approx} are less than $\tau$.  Thus
\eqref{eq:sg-approx} holds with the prescribed $\varepsilon$, and in
particular $\sup_{K_{\ell_+}}|H_j-U|<\tau$.  Set $H_\theta=H_j$ and
$W_\theta=W_j$, and select
\[
 \ell_-+\tau<\gamma<\ell_+-\tau.
\]
Then $H_\theta<\gamma$ on $K_{\ell_-}$ and
$H_\theta>\gamma$ on $\partial K_{\ell_+}$.  With
$c_\theta=\tanh\gamma$, monotonicity of tanh gives
$\{W_\theta\le c_\theta\}=\{H_\theta\le\gamma\}$.

Every target sublevel is path connected to the origin along trajectories.
Thus $K_{\ell_-}$ belongs to the origin component
$\Omega_{\theta,\gamma}^0$.  Since $H_\theta$ is proper, this component is
compact; it cannot leave $K_{\ell_+}$ without meeting its boundary, where
$H_\theta>\gamma$.  This proves \eqref{eq:sg-sandwich}.  Strict decrease on
the larger target sublevel implies that a trajectory starting in
$\Omega_{\theta,\gamma}^0$ cannot cross its boundary level $H_\theta=\gamma$
and, by continuity, remains in the same connected component. Thus
$\Omega_{\theta,\gamma}^0$ is positively invariant. Its compactness gives
forward completeness, and the Lyapunov conditions guarantee convergence to
the origin.

Finally choose levels
$\ell_-^j<\ell_+^j<\ell_-^{j+1}$ with $\ell_-^j\to\infty$. For each $j$,
choose a corresponding network $N_{\theta_j}$ and level $\gamma_j$, and write
$\Omega_j=\Omega_{\theta_j,\gamma_j}^0$. Then the components satisfy
$\Omega_j\subset K_{\ell_+^j}\subset K_{\ell_-^{j+1}}\subset\Omega_{j+1}$,
and their union is $\A$.
\end{proof}

\bibliographystyle{IEEEtran}
\bibliography{references}

\end{document}